\documentclass{article}

\usepackage{amsthm}
\usepackage{url}
\usepackage{authblk}
\usepackage{amsmath,amstext,amssymb,amsfonts,verbatim}

\usepackage{color}

\usepackage{graphicx,algorithmic,algorithm,rotating}
\allowdisplaybreaks[1]

\newtheorem{theorem}{Theorem}
\newtheorem{definition}[theorem]{Definition}
\newtheorem{proposition}[theorem]{Proposition}

\newtheorem{corollary}[theorem]{Corollary}

\begin{document}

\title{\vspace{-2cm}A Computable Measure of Algorithmic Probability by Finite Approximations with an Application to Integer Sequences}

\author[1,3]{Fernando Soler-Toscano}
\author[2,3,4]{Hector Zenil~\thanks{The author(s) declare(s) no conflict of interest regarding the publication of this paper. Corresponding author: hector.zenil [at] algorithmicnaturelab [dot] org}}
\affil[1]{Grupo de L\'ogica, Lenguaje e Informaci\'on, Universidad de Sevilla, Sevilla,
Spain.}   
\affil[2]{Information Dynamics Lab, Center for Molecular Medicine, Science For Life Laboratory (SciLifeLab), Department of Medicine Solna, Karolinska Institute, Stockholm, Sweden.}
\affil[3]{Group of Structural Biology, Department of Computer Science, University of Oxford, Oxford, U.K.}
\affil[4]{Algorithmic Nature Group, LABORES, Paris, France.}

\date{}

\renewcommand\Authands{ and }

\maketitle{}

\begin{abstract}
Given the widespread use of lossless compression algorithms to approximate algorithmic (Kolmogorov-Chaitin) complexity, and that, usually, generic lossless compression algorithms fall short at characterizing features other than statistical ones not different to entropy evaluations, here we explore an alternative and complementary approach. We study formal properties of a Levin-inspired measure $m$ calculated from the output distribution of small Turing machines. We introduce and justify finite approximations $m_k$ that have been used in some applications as an alternative to lossless compression algorithms for approximating algorithmic (Kolmogorov-Chaitin) complexity. We provide proofs of the relevant properties of both $m$ and $m_k$ and compare them to Levin's Universal Distribution. We provide error estimations of $m_k$ with respect to $m$. Finally, we present an application to integer sequences from the Online Encyclopedia of Integer Sequences which suggests that our AP-based measures may characterize non-statistical patterns, and we report interesting correlations with textual, function and program description lengths of the said sequences.\\

\noindent{}\textbf{Keywords:} Kolmogorov-Chaitin complexity, algorithmic probability, algorithmic Coding Theorem, Turing machines, integer sequences
\end{abstract}

\section{Algorithmic information measures}
\label{sec:introduction}

Central to Algorithmic Information Theory is the definition of algorithmic (Kol\-mo\-go\-rov-Chaitin or program-size)
complexity~\cite{kolmo,chaitin}: 
\begin{equation}
\label{kolmoeq}
K_T(s) = \min \{|p|, T(p)=s\}
\end{equation}
where $p$ is a program that outputs $s$ running on a universal Turing
machine $T$ and $|p|$ is the length in bits of $p$. The measure
was first conceived to define randomness and is today the accepted
objective mathematical measure of randomness, among other reasons 
because it has been proven to be mathematically
robust~\cite{levin}. In the following, 
we use $K(s)$ instead of $K_T(s)$ because the choice of $T$ is only
relevant up to an additive constant (Invariance Theorem). A technical 
inconvenience of $K$ as a function taking $s$  
to be the length of the shortest program that produces $s$ is its 
uncomputability. In other words, there is no program which takes a
string $s$ as input and produces the integer $K(s)$ as output. This is
 usually considered a major problem, but one ought to expect a
universal measure of randomness to have such a property. 

In previous papers~\cite{delahayezenil,d5} we have introduced a novel 
method to approximate $K$ based on the seminal concept of 
algorithmic probability (or AP), introduced by Solomonoff~\cite{solomonoff} and 
further formalized by Levin~\cite{levin}, who proposed the concept of uncomputable semi-measures and the so-called Universal Distribution. 

Levin's semi-measure\footnote{It is called a \emph{semi} measure because, unlike probability measures, the sum is never 1. This is due to the
 Turing machines that never halt.} $\mathfrak{m}_T$ defines the 
so-called Universal Distribution~\cite{kircher}, the value
$\mathfrak{m}_T(s)$ being the probability that a random program halts and produces $s$ running on a universal Turing machine $T$. The choice of $T$ is only relevant up to a multiplicative constant, so we will simply write 
$\mathfrak{m}$ instead of $\mathfrak{m}_T$. 

It is possible to use $\mathfrak{m}(s)$ to approximate $K(s)$ by 
means of the following theorem:

\begin{theorem}[Algorithmic Coding Theorem~\cite{levin}]\label{theo:codTh}
  There is a constant $c$ such that
  \begin{equation}\label{eq:CT}
|-\log_2 \mathfrak{m}(s) - K(s)| < c
  \end{equation}
\end{theorem}

This implies that if a string $s$ has many descriptions (high value of $\mathfrak{m}(s)$, as the string is produced many times, which implies a low value of $-\log_2 \mathfrak{m}(s)$, given that $\mathfrak{m}(s)<1$), it also has a short description (low value of $K(s)$). This is because the most frequent strings produced by programs of length $n$ are those which were already produced by programs of length $n-1$, as extra bits can produce redundancy in an exponential number of ways. On the other hand, strings produced by programs of length $n$ that could not be produced by programs of length $n-1$ are less frequently produced by programs of length $n$, as only very specific programs can generate them (see Section 14.6 in~\cite{cover}).  
This theorem elegantly connects probability to complexity---the frequency (or probability) of occurrence of a string with its algorithmic (Kolmogorov-Chaitin) complexity. It implies that~\cite{delahayezenil} one can calculate the Kolmogorov
complexity of a string from its
frequency~\cite{delahayezenil}, simply 
rewriting the formula as: 
\begin{equation}
K(s)=-\log_2 \mathfrak{m}(s) + O(1)
\end{equation}Thanks to this elegant connection established by ~\eqref{eq:CT} between algorithmic complexity and probability, our method can attempt to approximate an algorithmic probability measure by means of finite approximations using a fixed model of computation. The method is called the Coding Theorem Method (CTM)~\cite{d5}. 

In this paper we introduce $m$, a computable approximation to
$\mathfrak{m}$ that can be used to approximate $K$ by means of the
algorithmic Coding theorem. Computing $m(s)$ requires the output of a
numerable infinite number of Turing machines, so we first undertake
the investigation of finite approximations $m_k(s)$ that require only
the output of machines up to $k$ states. A key property of
$\mathfrak{m}$ and $K$ is their universality: the choice of the Turing machine used to compute the distribution is only relevant up to an (additive) constant, independent of the objects. The computability of this measure implies its lack of universality. The same is true when using common lossless compression algorithms to approximate $K$, but on top of their non-universality in the algorithmic sense, they are block entropy estimators as they traverse files in search of repeated patterns in a fixed-length window to build a replacement dictionary. Nevertheless, this does not prevent lossless compression algorithms to find useful applications in the same way as more algorithmic-based motivated measures can contribute even if also limited. Indeed,  $\mathfrak{m}$ has found successful 
applications in cognitive sciences~\cite{ploscompbio,nicoetal,gauvrit,visual,cognition}, in financial time series research~\cite{zenileco,linma}, graph theory and networks~\cite{zenilgraph,methodszenil,zenilphysres}. However, a thorough investigation to explore the properties of these measures, and to provide theoretical error estimations was missing. 

We start by presenting our Turing machine formalism (Section~\ref{sec:turing-machine-form}) and then show that it can be used to encode a prefix-free set of programs (Section~\ref{sec:turing-machines-as}). Then, in 
Section~\ref{sec:appr-levins-distr} we define a computable algorithmic probability measure $m$ based on our Turing machine formalism and prove its main properties, both for $m$ and for finite approximations $m_k$. In Section~\ref{sec:experiment} we compute $m_5$, compare it
with our previous distribution $D(5)$~\cite{d5} and estimate the error in $m_5$ as an approximation to $m$. We finish with some comments in
Section~\ref{sec:comm}.

\section{The Turing machine formalism}
\label{sec:turing-machine-form}

We denote by $(n,2)$ the class (or space) of all $n$-state 2-symbol Turing machines (with the halting state not included among the $n$ states) following the Busy Beaver Turing machine formalism as defined by Rado~\cite{rado}. Busy Beaver Turing machines are deterministic  machines with a single head and a single tape unbounded in both directions. When the machine enters the halting state the head no longer moves and the output is considered to comprise only the cells visited by the head prior to halting. Formally, 

\begin{definition}[Turing machine formalism] \label{def:turing-machine-form}
We designate as $(n,2)$ the \emph{set of Turing machines} with two symbols $\{0,1\}$ and $n$ states $\{1,\cdots,n\}$ plus a halting state $0$. These machines have $2n$ entries $(s_1,k_1)$ (for $s\in
\{1,\cdots,n\}$ and $k\in\{0,1\}$) in the \emph{transition table}, each with one \emph{instruction}
that determines their behavior. Such entries are
 represented by
\begin{equation}
\label{eq:instruction}
(s_1,k_1) \rightarrow (s_2, k_2,d)
\end{equation}
where $s_1$ and $k_1$ are 
respectively the current state and the symbol being read and 
$(s_2,k_2,d)$ represents the instruction to be executed: $s_2$ is the new state, $k_2$ the symbol to write and $d$ the direction. If $s_2$ is the halting state $0$, then $d=0$, otherwise $d$ is $1$ (right) or $-1$ (left).
\end{definition}

\begin{proposition} \label{prop:n2enumer}
  Machines in $(n,2)$ can be enumerated from $0$ to $(4n+2)^{2n}-1$  
\end{proposition}

\begin{proof}
Given the constraints in Definition~\ref{def:turing-machine-form}, for each transition of a 
Turing machine in $(n,2)$ there are $4n+2$ different instructions $(s_2,k_2,d)$. These are $2$ instructions when $s_2 = 0$ (given that $d=0$ is fixed and $k_2$ can be one of the two possible symbols) and $4n$ instructions if $s_2 \neq 0$ ($2$ possible moves, $n$ states and $2$ symbols). Then, considering the $2n$ entries in the transition table,
\begin{equation}
  \label{eq:1e}
  |(n,2)| = \left(4n+2\right)^{2n}
\end{equation}
These machines can be enumerated from $0$ to $|(n,2)|-1$. Several 
enumerations are possible. We can, for example, use a lexicographic ordering on the transitions~\eqref{eq:instruction}. 
\end{proof}

For the current paper, consider that some enumeration has been 
chosen. Thus we use $\tau^n_t$ to denote the machine number $t$ in $(n,2)$ following that enumeration.

\section{Turing machines as a prefix-free set of programs}
\label{sec:turing-machines-as}

We show in this section that the set of Turing machines following the 
Busy Beaver formalism can be encoded as a prefix-free set of programs 
capable of generating any finite non-empty binary string.

\begin{definition}[Execution of a Turing machine]
  \label{def:TMexec}
  Let $\tau\in(n,2)$ be a Turing machine. We denote by $\tau(i)$
 the execution of $\tau$ over an infinite tape filled with $i$ (a blank
  symbol), where $i\in\{0,1\}$. We write $\tau(i)\downarrow$ if
  $\tau(i)$ halts, and $\tau(i)\uparrow$ otherwise. We write
  $\tau(i)=s$ iff
  \begin{itemize}
  \item $\tau(i)\downarrow$, and
  \item $s$ is the output string of $\tau(i)$, defined as the
 concatenation of the symbols in the tape of $\tau$ that were visited at some instant of the execution $\tau(i)$.
  \end{itemize}
\end{definition}

As Definition~\ref{def:TMexec} establishes, we are only considering
machines running over a blank tape with no input. Observe that
 the output of $\tau(i)$ considers the symbols in all cells of the tape written on by $\tau$ during the computation, so the
  output contains the entire fragment of the tape that was used. To produce a 
symmetrical set of strings, we consider both symbols $0$ and $1$ as
possible blank symbols. 

\begin{definition}[Program]\label{deF:program}
  A \emph{program} $p$ is a triplet $\langle n, i, t\rangle$, where
  \begin{itemize}
  \item $n\geq 1$ is a natural number
  \item $i\in\{0,1\}$
  \item $0\leq t <(4n+2)^{2n}$
  \end{itemize}
  We say that the \emph{output} of $p$ is $s$ if, and only if, $\tau^n_t(i)=s$.
\end{definition}

Programs can be executed by a universal Turing machine that reads a 
binary encoding of $\langle n, i, t\rangle$ (Definition~\ref{deF:progrEncod})
and simulates $\tau^n_t(i)$. Trivially, for each finite binary 
string $s$ with length $|s|>0$, there is a program $p$ which outputs 
$s$. 

Now that we have a formal definition of programs, we show 
that the set of valid programs can be 
represented as a prefix-free set of binary strings.

\begin{definition}[Binary encoding of a program]\label{deF:progrEncod}
Let $p=\langle n, i, t\rangle$ be a program
  (Definition~\ref{deF:program}). The \emph{binary encoding} of $p$ is
 a binary string with the following sequence of bits:
  \begin{itemize}
  \item First, $1^{n-1}0$, that is, $n-1$ repetitions of $1$ followed
 by $0$. This way we encode $n$.
  \item Second, a bit with value $i$ encodes the blank symbol.
  \item Finally, $t$ is encoded using $\lceil \log_2
\left((4n+2)^{2n}\right) \rceil$ bits.
  \end{itemize}
\end{definition}

The use of $\lceil \log_2 \left((4n+2)^{2n}\right) \rceil$ bits to 
represent $t$ ensures that all programs with the same $n$ are 
represented by strings of equal size. As there are $(4n+2)^{2n}$
machines in $(n,2)$, with these bits we can represent any value of
$t$. The process of reading the binary encoding of a 
program $p=\langle n, i, t\rangle$ and simulating $\tau^n_t(i)$ is 
computable, given the enumeration of Turing machines. 

As an example, this is the binary representation of the program
$\langle 2, 0, 185\rangle$:
\begin{center}
  \begin{tabular}{|cc|c|cccccccccccccc|}
\hline 
1 & 0 & 0 & 0 & 0 & 0 & 0 & 0 & 0 & 1 & 0 & 1 & 1 & 1 & 0 & 0 & 1 \\\hline
  \end{tabular}
\end{center}

The proposed encoding is prefix-free, that is, there is no pair 
of programs $p$, $p'$ such that the binary 
encoding of $p$ is a prefix of the binary encoding of $p'$. This is because the $n$ initial bits 
of the binary encoding of $p=\langle n, i, t\rangle$ determine the 
length of the encoding. So $p'$ cannot be encoded by a 
binary string having a different length but the same $n$ initial 
bits.

\begin{proposition}[Programming by coin flips]\label{prop:coinFlips}
  Every source producing an arbitrary number of random bits generates
 a unique program (provided it generates at least one $0$).
\end{proposition}

\begin{proof}
  The bits in the sequence are used to produce a unique program
 following Definition~\ref{deF:progrEncod}. We start by producing the
 first $n$ part, by selecting all bits until the first $0$
 appears. Then the next bit gives $i$. Finally, as we know the value
 of $n$, we take the following $\lceil \log_2
  \left((4n+2)^{2n}\right) \rceil$ bits to set the value of $t$. It is
 possible that constructing the program in this way, the value of $t$
 is greater than the maximum $(4n+2)^{2n}-1$ in the enumeration. In
 which case we associate the program with some trivial non-halting Turing
 machine. For example a machine with the initial transition 
staying at the initial state. 
\end{proof}

The idea of programming by coin flips is very common in
 Algorithmic Information Theory. It produces a prefix-free coding
  system, that is, there is no string $w$ encoding a program $p$ which
 is a prefix of a string $wz$ encoding a program $p' \neq p$. These
 coding systems make longer programs (for us, Turing machines with
 more states) exponentially less probable than short programs. In our
 case, this is because of the initial sequence of $n-1$ repetitions
 of $1$, which are produced with probability $1/2^{n-1}$.  This
 observation is important because when we later use machines in
  $\bigcup_{n=1}^k (n,2)$ to reach a finite approximation of our
  measure, the greater $k$ is, the exponentially smaller the error we
  will be allowing: the probability of producing by coin flips a
  random Turing machine with more than $k$ states decreases
  exponentially with $k$~\cite{cover}.

\section{A Levin-style algorithmic measure}
\label{sec:appr-levins-distr}

\begin{definition}
  Given a Turing machine $\mathcal{A}$ accepting a prefix-free set of
 programs, the \emph{probability distribution} of $\mathcal{A}$ is
 defined as
  \begin{equation}
\label{eq:12}
P_\mathcal{A}(s) = \sum_{p: \mathcal{A}(p) = s} \frac{1}{2^{|p|}}
  \end{equation}
  where $\mathcal{A}(p)$ is equal to $s$ if and only if $\mathcal{A}$
  halts with input $p$ and produces $s$. The length in bits of program
  $p$ is represented by $|p|$. 
\end{definition}

If $\mathcal{A}$ is a universal Turing machine,
$P_{\mathcal{A}}(s)$ measures how frequently the output $s$
 is generated when running random programs at $\mathcal{A}$. Given that the sum of $P_{\mathcal{A}}(s)$ for all strings is not $1$ (non-halting programs
 not producing any strings are counted in $2^{|p|}$) it is said to be a semi-measure, also known as Levin's distribution~\cite{levin}. The distribution is universal in the sense that the choice of $\mathcal{A}$ (among all the infinite possible universal reference Turing machines) is only relevant up to a multiplicative  constant and that the distribution is based on the universal model of Turing computability.

\begin{definition}[Distribution $m(s)$]
  Let $\mathcal{M}$ be a Turing machine executing the programs introduced in
 Definition~\ref{deF:program}. Then, $m(s)$ is defined by
  \[m(s) = P_{\mathcal{M}}(s).\]
\end{definition}

\begin{theorem}For any binary string $s$, 
  \begin{equation}
\label{eq:2}
m(s) = \sum_{n=1}^\infty \frac{|\{ \tau\in(n,2) \mid \tau(0)=s
  \}|+|\{  \tau\in(n,2) \mid \tau(1)=s  \}|}{2^{n+1+\lceil \log_2
\left((4n+2)^{2n}\right) \rceil}} 
  \end{equation}
\end{theorem}

\begin{proof}
  By Definition~\ref{deF:progrEncod}, the length of the encoding of
 program $p=\langle n, i, t\rangle$ is $n+1+\lceil \log_2
  \left((4n+2)^{2n}\right) \rceil$. It justifies the denominator
  of~\eqref{eq:2}, as~\eqref{eq:12} requires it to be $2^{|p|}$. For
  the numerator, observe that the set of programs producing $s$ with
 the same $n$ value corresponds to all machines in $(n,2)$ producing
  $s$ with either $0$ or $1$ as blank symbol. Note that if a machine
  produces $s$ both with $0$ and $1$, it is counted twice, as each execution is represented by a different program (that differ only
 as to the $i$ digit).
\end{proof}

\subsection{Finite approximations to $m$}
\label{sec:finite-appr-ms}

The value of $m(s)$ for any string $s$ depends on the output of an
 infinite set of Turing machines, so we have to manage ways to 
approximate it. The method proposed in Definition~\ref{def:mk}
approximates $m(s)$ by considering only a finite number of Turing 
machines up to a certain number of states. 

\begin{definition}[Finite approximation $m_k(s)$] \label{def:mk}
  The \emph{finite approximation to} $m(s)$ \emph{bound to} $k$
  \emph{states}, $m_k(s)$, is defined as 
  \begin{equation}
\label{eq:8}
m_k(s) = \sum_{n=1}^k \frac{|\{ \tau\in(n,2) \mid \tau(0)=s
  \}|+|\{  \tau\in(n,2) \mid \tau(1)=s  \}|}{2^{n+1+\lceil \log_2
\left((4n+2)^{2n}\right) \rceil}}
  \end{equation}  
\end{definition}

\begin{proposition}[Convergence of $m_k(s)$ to
  $m(s)$]\label{prop:convergenve} 
  \[\sum_{s\in(0+1)^\star} \mid m(s) - m_k(s) \mid \ \ \leq
\frac{1}{2^{k}} 
  \]
\end{proposition}

\begin{proof}
  By~\eqref{eq:2} and~\eqref{eq:8},
  \begin{eqnarray*}
\sum_{s\in(0+1)^\star} \mid m(s) - m_k(s) \mid & = & 
\sum_{s\in(0+1)^\star} m(s) \ \ - \sum_{s\in(0+1)^\star} m_k(s)
 \\
& \leq  & \sum_{n=k+1}^\infty \frac{2 (4n+2)^{2n}}{2^{n+1+\lceil \log_2
\left((4n+2)^{2n}\right) \rceil}} \\ & \leq &
\sum_{n=k+1}^\infty \frac{2 (4n+2)^{2n}}{2^n \cdot 2 \cdot 2^{\log_2
\left((4n+2)^{2n}\right)}}\\
& = & \sum_{n=k+1}^\infty \frac{1}{2^n} = \frac{1}{2^k}
  \end{eqnarray*}
\end{proof}

Proposition~\ref{prop:convergenve} ensures that the sum of the error
 in $m_k(s)$ as an approximation to $m(s)$, for all strings $s$, 
decreases exponentially with $k$. The question of this convergence was first broached in~\cite{delahaye2007}. The bound of $1/2^k$ has only 
theoretical value; in practice we can find lower bounds. In fact, the 
proof counts all $2(4n+2)^{2n}$ programs of size $n$ to bound the 
error (and many of them do not halt). In 
Section~\ref{sec:error-calculation} we provide a finer error 
calculation for $m_5$ by removing from the count some very trivial 
machines that do not halt. 

\subsection{Properties of $m$ and $m_k$}
\label{sec:properties-ms-m_ks}

Levin's distribution is characterized by some important 
properties. First, it is lower semi-computable, that is, it is 
possible to compute lower bounds for it. Also, it is a semi-measure, 
because the sum of 
probabilities for all strings is smaller than $1$. The key property of 
Levin's distribution is 
its \textit{universality}: a semi-measure $P$ is universal if and only if for every 
other semi-measure $P'$ there exists a constant $c>0$ (that may depend only on $P$ and $P'$) such that for every string $s$, $c\cdot P(s)\geq P'(s)$. That is, a distribution is universal if and only if it dominates (modulo a multiplicative constant) every other semi-measure. In this section we present some results pertaining to the computational properties of $m$ and $m_k$. 

\begin{proposition}[Runtime bound]\label{prop:runtimeBound}
Given any binary string $s$, a machine with $k$ states producing $s$ runs a maximum of $2^{|s|} \cdot |s| \cdot k$ steps
 upon halting or never halts.
\end{proposition}

\begin{proof}
 Suppose that a machine $\tau$ produces $s$. We can trace back the computation of $\tau$ upon halting by looking at the portion of $|s|$ cells in the tape that will constitute the output. Before each step, the machine may be in one of $k$ possible states, reading one of the $|s|$ cells. Also, the $|s|$ cells can be filled in $2^{|s|}$ ways (with a $0$ or $1$ in each cell). This makes for $2^{|s|} \cdot|s| \cdot k$ different possible instantaneous descriptions of the computation. So any machine may run, at most, that number of steps in order to produce $s$. Otherwise, it would produce a string with a greater length (visiting more than $|s|$ cells) or enter a loop.

\end{proof}

Observe that a key property of our output convention is that we use
 all visited cells in the machine tape. This is what gives us the 
runtime bound which serves to prove the most important property of
$m_k$, its computability (Theorem~\ref{prop:computMK}). 

\begin{theorem}[Computability of $m_k$]\label{prop:computMK}
  Given $k$ and $s$, the value of $m_k(s)$ is computable.
\end{theorem}

\begin{proof}
  According to~\eqref{eq:8} and Proposition~\ref{prop:n2enumer}, there is a finite number of machines involved in the computation of $m_k(s)$. Also, Proposition~\ref{prop:runtimeBound} sets the maximum runtime for any of these machines in order to produce $s$. So an algorithm to compute $m_k(s)$ enumerates all machines in $(n,2)$,  $1\leq n\leq k$ and runs each machine to the corresponding bound. 
\end{proof}

\begin{corollary}\label{coro:computableMini}
  Given a binary string $s$, the minimum $k$ with $m_k(s)>0$
 is computable. 
\end{corollary}

\begin{proof}
  Trivially, $s$ can be produced by a Turing machine with $|s|$
  states in just $s$ steps. At each step $i$, this machine writes the
  $i$\textsuperscript{th} symbol of $s$, moves to the right and
 changes to a new state. When all symbols of $s$ have been written,
 the machine halts. So, to get the minimum $k$ with $m_k(s)>0$, we
 can enumerate all machines in $(n,2)$, $1\leq n\leq |s|$ and run all
  of them up to the runtime bound given by
  Proposition~\ref{prop:runtimeBound}. The first machine producing $s$
  (if the machines are enumerated from smaller to larger size) gives the
  value of $k$. 
\end{proof}

Now, some uncomputability results of $m_k$

\begin{proposition}\label{prop:uncompLongest}
  Given $k$, the length of the longest $s$ with $m_k(s)>0$ is
  non-computable. 
\end{proposition}

\begin{proof}
  We proceed by contradiction. Suppose that such a computable function as $l(k)$ gives the length of the longest $s$ with $m_k(s)>0$. Then ?$l(k)$, together with the runtime bound in Proposition~\ref{prop:runtimeBound}, provides a computable function that gives the maximum runtime that a machine in $(k,2)$ may run prior to halting. But it contradicts the uncomputability of the Busy
 Beaver~\cite{rado}: the highest runtime of halting machines in
 $(k,2)$ grows faster than any computable function. 
\end{proof}

\begin{corollary}\label{prop:uncompHowMany}
  Given $k$, the number of different strings $s$ with $m_k(s)>0$ is
 non-computable.
\end{corollary}

\begin{proof}
  Also by contradiction: If the number of different strings with
  $m_k(s)>0$ is computable, we can run in parallel all machines in
 $(k,2)$ until the corresponding number of different strings has
 been found. This gives us the longest string, which is in contradiction to
 Proposition~\ref{prop:uncompLongest}. 
\end{proof}

Now to the key property of $m$, its computability,

\begin{theorem}[Computability of $m$]\label{teo:computM}
Given any non-empty binary string, $m(s)$ is computable. 
\end{theorem}

\begin{proof}
  As we argued in the proof of Corollary~\ref{coro:computableMini}, a
 non-empty binary string $s$ can be produced by a machine with $|s|$
 states. Trivially, it is then also produced by machines with more
 than $|s|$ states. So for every non-empty string $s$, the value of $m(s)$, according to~\eqref{eq:2}, is the sum of enumerable infinite many rationals which produce a real number. A real number is
 computable if, and only if, there is some algorithm that, given $n$,
 returns the first $n$ digits of the number. And this is what $m_k(s)$
  does. Proposition~\ref{prop:convergenve} enables us to calculate the
 value of $k$ such that $m_k(s)$ provides the required digits of
 $m(s)$, as $m(s)-m_k(s)$ is bounded by $1/2^k$. 
\end{proof}

The subunitarity of $m$ and $m_k$ implies that the sum of $m(s)$
(or $m_k(s)$) for all strings $s$ is smaller than one. This is 
because of the non-halting machines: 

\begin{proposition}[Subunitarity] \label{prop:subunit}
  The sum of $m(s)$ for all strings
  $s$ is smaller than $1$, that is,
  \[
  \sum_{s\in (0+1)^\star} m(s) < 1
  \]
\end{proposition}

\begin{proof}
  By using~\eqref{eq:2},
  \begin{equation}
\label{eq:4}
\sum_{s\in (0+1)^\star} m(s) = \sum_{n=1}^\infty \frac{|\{
  \tau\in(n,2) \mid \tau(0)\downarrow
  \}|+|\{  \tau\in(n,2) \mid \tau(1)\downarrow  \}|}{2^{n+1+\lceil \log_2
\left((4n+2)^{2n}\right) \rceil}} 
  \end{equation}
  but $|\{\tau\in(n,2) \mid \tau(0)\downarrow \}|+|\{  \tau\in(n,2)
  \mid \tau(1)\downarrow  \}|$ is the number of machines in $(n,2)$
  that halt when starting with a blank tape filled with $0$ plus the number
  of machines in $(n,2)$ that halt when starting on a blank tape
 filled with $1$. This number is at most twice the cardinality of
 $(n,2)$, but we know that it is smaller, as there are very trivial
 machines that do not halt, such as those without transitions to the
  halting state, so
  \begin{eqnarray*}
\sum_{s\in (0+1)^\star} m(s) & < & \sum_{n=1}^\infty \frac{2
  (4n+2)^{2n}}{2^{n+1+\lceil \log_2 
\left((4n+2)^{2n}\right) \rceil}} = \sum_{n=1}^\infty
\frac{(4n+2)^{2n}}{2^n\cdot 2^{\lceil \log_2 
\left((4n+2)^{2n}\right) \rceil}} \\
& \leq & \sum_{n=1}^\infty \frac{(4n+2)^{2n}}{2^n
(4n+2)^{2n}} = \sum_{n=1}^\infty \frac{1}{2^n} = 1
  \end{eqnarray*}
\end{proof}

\begin{corollary}
  The sum of $m_k(s)$ for all strings $s$ is smaller than $1$
\end{corollary}

\begin{proof}
  By Proposition~\ref{prop:subunit}, \eqref{eq:2} and~\eqref{eq:8}.
\end{proof}

The key property of $m_k(s)$ and $m(s)$ is their computability, given 
by Propositions~\ref{prop:computMK} and~\ref{teo:computM},
respectively. So these distributions cannot be universal, as Levin's 
Universal Distribution is non-computable. In spite of this, the 
computability of our distributions (and the possibility of
 approximating them with a reasonable computational effort), as we have shown, provides us with a tool to approximate the algorithmic probability of short binary strings. In some sense this is 
similar to what happens with other (computable) approximations to
(uncomputable) Kolmogorov complexity, such as common lossless compression algorithms, which in turn are estimators of the classical Shannon entropy rate (e.g. all those based in LZW, and unlike $m_k(s)$ and $m(s)$, are not able to find algorithmic content beyond statistical patterns, not even in principle, unless a compression algorithm is designed to seek a specific one. For example, the digital 
expansion of the mathematical constant $\pi$ is believed to be normal and therefore will contain no statistical patterns of the kind that
compression algorithms can detect, yet there will be a (short)
computer program that can generate it, or at least finite (and small)
initial segments of $\pi$.

\section{Computing $m_5$}
\label{sec:experiment}

We have explored the sets of Turing machines in $(n,2)$ for $n\leq 5$
in previous papers~\cite{delahayezenil,d5}. For $n\leq 4$, the maximum time that a machine
 in $(n,2)$ may run upon hating is known~\cite{brady}. It allows us to 
calculate the exact values of $m_4$. For $n=5$, we have
 estimated~\cite{d5} that 500 steps cover almost the totality of 
halting machines. We have the database of machines producing each 
string $s$ for each value of $n$. So we have applied~\eqref{eq:8} to 
estimate $m_5$ (because we set a low runtime). 

In previous papers~\cite{d5,computability}, we worked with $D(k)$, a 
measure similar to $m_k$, but the denominator of~\eqref{eq:8} 
is the number of (detected) halting machines in $(k,2)$. Using 
$D(5)$ as an approximation to Levin's distribution, algorithmic complexity is 
estimated\footnote{These values can be consulted at
  \texttt{http://www.complexitycalculator.com}. Accessed on June 22, 2017.} by means of the
 algorithmic Coding Theorem~\ref{theo:codTh} as $K_{D(5)}(s) =
-\log_2 D(5)(s)$. Now, $m_5$ provides us with another estimation:
$K_{m_5}(s) = -\log_2 m_5(s)$. Table~\ref{tab:top10strings} shows the
10 most frequent strings in both distributions, together with their
estimated complexity.

\begin{table}[htbp!]
  \centering
  \begin{tabular}{ccc||ccc}
$s$ & $K_{m_5}(s)$ & $K_{D(5)}(s)$ & $s$ & $K_{m_5}(s)$ &
$K_{D(5)}(s)$ \\ \hline\hline 
0 &  3.7671 &  2.5143  &  11 &  6.8255 & 3.3274 \\
1 &  3.7671 &  2.5143 &  000 & 10.4042 &  5.3962 \\
00 & 6.8255 &  3.3274 &  111 & 10.4042 & 5.3962 \\
01 &  6.8255 & 3.3274 &  001 & 10.4264 & 5.4458 \\
10 &  6.8255 & 3.3274 &  011 & 10.4264  & 5.4458 \\
  \end{tabular}
  \caption{Top 10 strings in $m_5$ and $D(5)$ with their estimated
complexity} 
  \label{tab:top10strings}
\end{table}

\begin{figure}[htbp!]
  \centering
  \includegraphics[width=4.5cm]{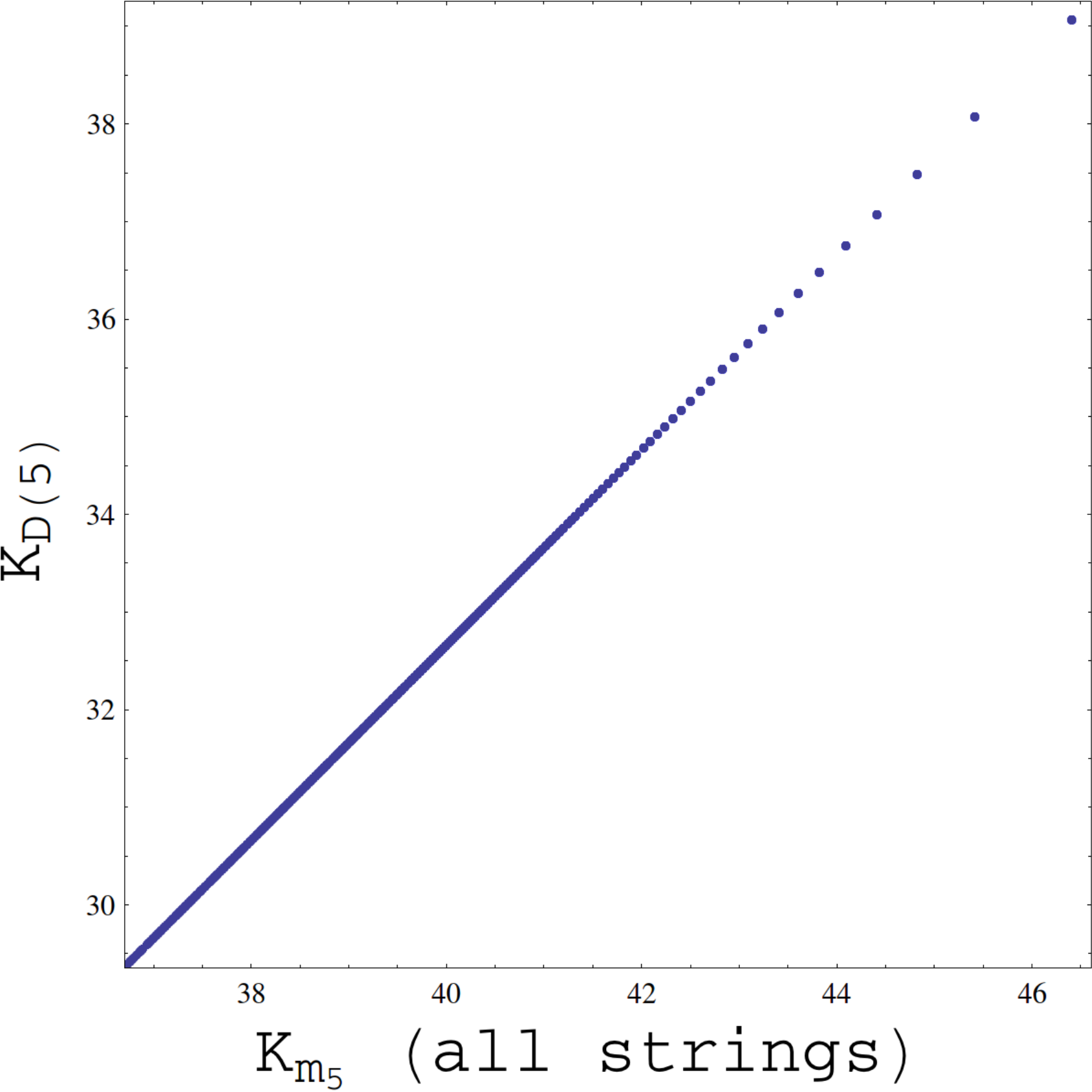}\hspace{1cm}
  \includegraphics[width=4.5cm]{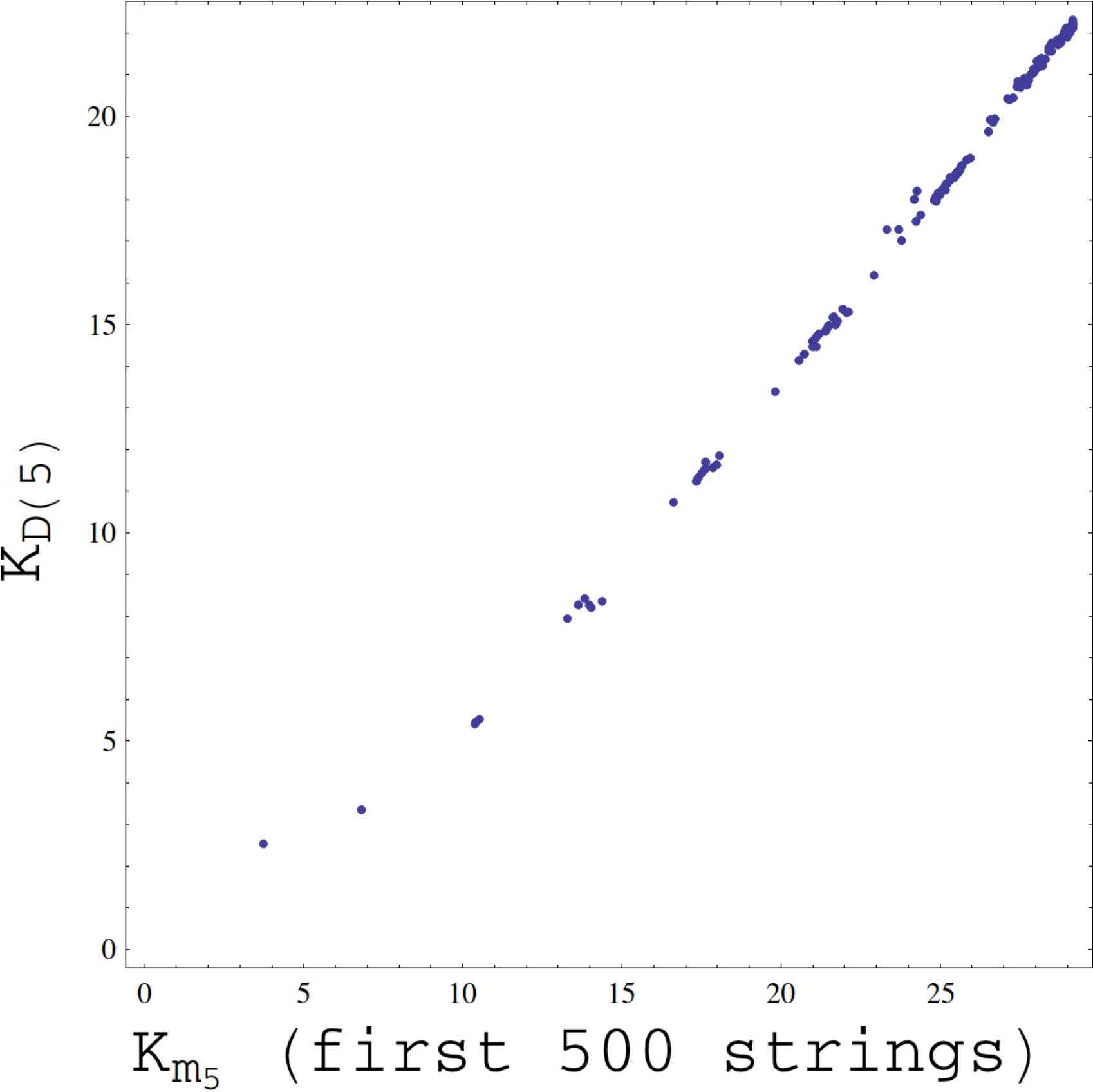}
  \caption{Correlation of rank comparison between $K_{m_5}$ and
$K_{D(5)}$} 
  \label{fig:compRank}
\end{figure}

Figure~\ref{fig:compRank} shows a rank comparison of both estimations
of algorithmic complexity after application of the algorithmic Coding
 Theorem. With minor differences, there is an almost  
perfect agreement. So in classifying strings according to their relative 
algorithmic complexity, the two distributions are equivalent. 

The main difference between $m_k$ and $D(k)$ is that $D(k)$ is not
 computable, because computing it would require us to know the exact number of halting 
machines in $(k,2)$, which is impossible given the halting 
problem. We work with approximations to $D(k)$ by considering the 
number of halting machines detected. In any case, though $m_k$ is 
computable, it is computationally intractable, so in practice
 (approximations to) the two measures can be used interchangeably .

\subsection{Error calculation}
\label{sec:error-calculation}

We can make some estimations about the error in $m_5$ with respect to
$m$.  ``0'' and ``1'' are two very special strings, both with 
the maximum $m_5$ value. These strings are the most frequent outputs in 
$(n,2)$ for $n\leq 5$, and we may conjecture that they are the most 
frequent outputs for all values of $n$. These strings then have the
 greatest absolute error, because the terms in the sum of 
$m(\text{``0''})$ (the argument for $m(\text{``1''})$ is identical) 
not included in $m_5(\text{``0''})$ are always the greatest 
independent of $n$. 

We can calculate the exact value of the terms for
$m(\text{``0''})$ in~\eqref{eq:2}. To 
produce ``0'', starting with a tape filled with $i\in\{0,1\}$, a machine
in $(n,2)$ must have the transition corresponding to the initial state 
and read symbol $i$ with the following instruction: write $0$ and
change to the halting state (thus not moving the head). The other $2n-1$
transitions may have any of the $4n+2$ possible instructions. So there
 are $(4n+2)^{2n-1}$ machines in $(n,2)$ producing ``0'' when running 
on a tape filled with $i$. Considering both values of $i$, we have
$2(4n+2)^{2n-1}$ programs of the same length $n+1+\lceil \log_2
\left((4n+2)^{2n}\right) \rceil$ producing ``0''. Then, for ``0'',
\begin{equation}
  \label{eq:3}
  m(\text{``0''}) = \sum_{n=1}^\infty \frac{2(4n+2)^{2n-1}}{2^{n+1+\lceil \log_2
\left((4n+2)^{2n}\right) \rceil}}
\end{equation}
This can be approximated by
\begin{eqnarray*}
  m(\text{``0''}) & = & \sum_{n=1}^\infty
  \frac{2(4n+2)^{2n-1}}{2^{n+1+\lceil \log_2 \left((4n+2)^{2n}\right)
  \rceil}} \\  
  & = & 
  \sum_{n=1}^\infty
  \frac{2(4n+2)^{2n-1}}{2^{n+1} 2^{\lceil \log_2 \left((4n+2)^{2n}\right)
  \rceil}}\\
  & = & 
   \sum_{n=1}^\infty \frac{(4n+2)^{2n-1}}{2^n 2^{\lceil \log_2 \left((4n+2)^{2n}\right)
  \rceil}}\\
  & = & 
   \sum_{n=1}^{2000} \frac{(4n+2)^{2n-1}}{2^n 2^{\lceil \log_2 \left((4n+2)^{2n}\right)
  \rceil}} +
  \sum_{n=2001}^{\infty} \frac{(4n+2)^{2n-1}}{2^n 2^{\lceil \log_2 \left((4n+2)^{2n}\right)
  \rceil}}
  \\  
  & < &  
   \sum_{n=1}^{2000} \frac{(4n+2)^{2n-1}}{2^n 2^{\lceil \log_2 \left((4n+2)^{2n}\right)
  \rceil}} +
  \sum_{n=2001}^\infty \frac{(4n+2)^{2n-1}}{2^n 2^{\log_2
 \left((4n+2)^{2n}\right)}} \\ 
 & = & 
   \sum_{n=1}^{2000} \frac{(4n+2)^{2n-1}}{2^n 2^{\lceil \log_2 \left((4n+2)^{2n}\right)
   \rceil}} +
   \sum_{n=2001}^\infty \frac{(4n+2)^{2n-1}}{2^n (4n+2)^{2n}} \\
   & =  & 
   \sum_{n=1}^{2000} \frac{(4n+2)^{2n-1}}{2^n 2^{\lceil \log_2 \left((4n+2)^{2n}\right)
   \rceil}} +
   \sum_{n=2001}^\infty \frac{1}{2^n (4n+2)} \simeq 0.0742024
\end{eqnarray*}
we have divided the infinite sum into two intervals cutting at 2000 
because the approximation of $2^{\lceil \log_2 \left((4n+2)^{2n}\right)
  \rceil}$ to $(4n+2)^{2n}$ is not good for low values of $n$, but 
has almost no impact for large $n$. In fact, cutting at 1000 or 
4000 gives the same result with a precision of 17 decimal
 places. We have used \emph{Mathematica} to calculate both the sum
 from $1$ to $2000$ and the convergence from $2001$ to infinity. So 
the value $m(\text{``0''}) = 0.0742024$ is exact for practical purposes. The value of $m_5(\text{``0''})$ is $0.0734475$, so the error in the 
calculation of $m(\text{``0''})$ is $0.0007549$. If ``0'' and ``1'' 
are the strings with the highest $m$ value, as we (informedly) conjecture, then this is the maximum error in $m_5$ as an approximation to $m$. 

As a reference, $K_{m_5}(\text{``0''})$ 
is $3.76714$. With the real $m(\text{``0''})$ value, the approximated complexity is $3.75239$. The difference is not relevant for most 
practical purposes. 

We can also provide an upper bound for the sum of the error in $m_5$ for 
strings different from ``0'' and ``1''. Our way of proceeding is similar to the proof of Proposition~\ref{prop:convergenve}, but we count in a 
finer fashion. The sum of the error for strings different from ``0'' and
``1'' is 

{\footnotesize
  \begin{equation}
\label{eq:9}
  \sum_{n=6}^\infty 
  \frac{|\{ \tau\in(n,2) \mid \tau(0)\downarrow,\
\tau(0)\notin\{\text{``0''},\text{``1''}\} \}|+|\{ \tau\in(n,2)
\mid 
\tau(1)\downarrow,\
\tau(1)\notin\{\text{``0''},\text{``1''}\}
 \}|}{2^{n+1+\lceil \log_2
   \left((4n+2)^{2n}\right) \rceil}} 
  \end{equation}
\normalsize{}}

The numerators of the above sum contain the number of computations
 (with blank symbol ``0'' or ``1'') of Turing machines in $(n,2)$,
$n\geq 6$, that halt and produce an output different from ``0'' and
``1''. We can obtain an upper bound of this value by removing, from the 
set of computations in $(n,2)$, those that produce ``0'' or ``1'' and 
some trivial cases of machines that do not halt. 

First, the number of computations in $(n,2)$ is $2(4n+2)^{2n}$, as all
machines in $(n,2)$ are run twice for both blank symbols (``0" and
``1"). Also, the computations producing ``0'' or ``1'' are $4(4n+2)^{2n-1}$. Now, we
 focus on two sets of trivial non-halting machines:
\begin{itemize}
\item Machines with the initial transition staying at the initial
  state. For blank symbol $i$, there are $4(4n+2)^{2n-1}$ machines
 that when reading $i$ at the initial state do not change the state
 (for the initial transition there are $4$ possibilities, depending
  on the writing symbol and direction, and for the other $2n-1$
  transitions there are $4n+2$ possibilities). These machines will
  keep moving in the same direction without halting. Considering both
 blank symbols, we have $8(4n+2)^{2n-1}$ computations of this kind.
\item Machines without transition to the halting state. To keep the intersection of this and the above set empty, we also consider that the
 initial transition moves to a state different from the initial state. So
 for blank symbol $i$, we have $4(n-1)$ different initial transitions
  ($2$ directions, $2$ writing symbols and $n-1$ states) and $4n$
  different possibilities for the other $2n-1$ transitions. This makes
 a total of $4(n-1)(4n)^{2n-1}$ different machines for blank symbol
  $i$ and $8(n-1)(4n)^{2n-1}$ computations for both blank symbols. 
\end{itemize}
Now, an upper bound for~\eqref{eq:9} is:
\begin{displaymath}
  \sum_{n=6}^\infty 
  \frac{2(4n+2)^{2n} - 4(4n+2)^{2n-1} - 8(4n+2)^{2n-1} -
8(n-1)(4n)^{2n-1}}{2^{n+1+\lceil \log_2 
   \left((4n+2)^{2n}\right) \rceil}}
\end{displaymath}
The result of the above sum is $0.0104282$ (smaller than $1/32$, as
 guaranteed by Proposition~\ref{prop:convergenve}). This is an upper bound of 
the sum of the error $m(s)-m_5(s)$ for all infinite strings $s$ 
different from ``0'' and ``1''. Smaller upper bounds can be found by 
removing from the above sum other kinds of predictable non-halting 
machines.

\section{Algorithmic Complexity of Integer Sequences}

Measures that we introduced based on finite approximations of algorithmic probability have found applications in areas ranging from economics~\cite{zenileco} to human behavior and cognition~\cite{visual, cognition,ploscompbio} to graph theory~\cite{zenilgraph}. We have explored the use of other models of computation suggesting similar and correlated results in output distribution~\cite{zenilalgo} and compatibility, in a range of applications, with general compression algorithms~\cite{computability,kolmo2d}. We also investigated~\cite{d5}  the behavior of the additive constant involved in the Invariance theorem from finite approximations to $D(5)$, strongly suggesting fast convergence and smooth behavior of the invariance constant. In~\cite{zenilgraph} and~\cite{kolmo2d}, we introduced an AP-based measure for 2-dimensional patterns, based on replacing the tape of the reference Turing machine for a 2-dimensional grid. The actual implementation requires breaking any grid into smaller blocks for which we then have estimations of their algorithmic probability according to the Turing machine formalism described in~\cite{bdmpaper,zenilgraph} and~\cite{kolmo2d}.

Here we introduce an application of AP-based measures--as described above--to integer sequences. We show that an AP-based measure constitutes an alternative or complementary tool to lossless compression algorithms, widely used to find estimations of algorithmic complexity.

\subsection{AP-based measure}

The AP-based method used here is based on the distribution $D(5)$ and is defined just like $m_k(s)$. However, to increase its range of applicability, given that $D(5)$ produces all $2^{12}$ bit-strings of length 12 except for 2 (that are assigned maximum values and thus complete the set), we introduce what we call the \textit{Block Decomposition Method} (BDM) that decomposes strings longer than 12 into strings of maximum length 12 that can be derived from $D(5)$. The final estimation of the complexity of a string longer than 12 bits is then the result of the sum of the complexities of the different substrings of length not exceeding 12 in $D(5)$ if they are different, but the sum of only $log_2(n)$ if $n$ substrings are the same. The formula is motivated by the fact that $n$ strings that are the same do not have $n$ times the complexity of one of the strings but rather $\log_2(n)$ times the complexity of just one of the substrings. This is because the algorithmic complexity of the $n$ substrings to be considered is the length of at most the `print($s$) $n$ times' program and not the length of `print($ss\ldots s$)'. We have shown that this measure is a hybrid measure of complexity, providing local estimations of algorithmic complexity and global evaluations of Shannon entropy~\cite{bdmpaper}. Formally,

$$BDM(X) = \sum_i^k m_5(x_i) + \log(s_i)
$$

where $s_i$ is the multiplicity of $x_i$, and $x_i$ the subsequences from the decomposition of $X$ into $k$ subsequences, with a possible remainder sequence $y<x$ if $|X|$ is not a multiple of the decomposition length $l$. 
More details on error estimations for this particular measure extending the power of $m_5$, and on the boundary conditions, are given in~\cite{bdmpaper}.

\subsection{The On-Line Encyclopedia of Integer Sequences (OEIS)}

The On-Line Encyclopedia of Integer Sequences (OEIS) is a database with the largest collection of integer sequences. It is created and maintained by Neil Sloane and the OEIS Foundation. 

Widely cited, the OEIS stores information on integer sequences of interest to both professional mathematicians and amateurs. As of 30 December 2016 it contained nearly 280\,000 sequences, making it the largest database of its kind. 

We found 875 binary sequences in the OEIS database, accessed through the knowledge engine WolframAlpha Pro and downloaded with the Wolfram Language. 

Examples of descriptions found to have the greatest algorithmic probability include the sequence ``A maximally unpredictable sequence" with associated sequence 0 1 0 0 1 1 0 1 0 1 1 1 0 0 0 1 0 0 0 0 1 1 1 1 0 1 1 0 0 1 0 1 0 0 1 0 0 1 1 1; or A068426, the ``Expansion of $ln_2$ in base 2" and associated sequence 0 1 0 0 0 1 1 0 1 1 0 0 0 0 0 1 0 1 0 0 1 1 1 0 0 1 0 1 1 1 0 1 1 1 0 0 0 0 0 0. This contrasts with sequences of high entropy such as sequence A130198, the single paradiddle, a four-note drumming pattern consisting of two alternating notes followed by two notes with the same hand, with sequence 0 1 0 0 1 0 1 1 0 1 0 0 1 0 1 1 0 1 0 0 1 0 1 1 0 1 0 0 1 0 1 1 0 1 0 0 1 0 1 1, or sequence A108737, found to be among the less compressible, with the description ``Start with $S = {}$. For $m = 0, 1, 2, 3, \ldots$ let $u$ be the binary expansion of $m$. If $u$ is not a substring of $S$, append the minimal number of 0s and 1s to $S$ to remedy this. The sequence gives $S$" and sequence 0 1 0 1 1 0 0 1 1 1 0 0 0 1 0 1 0 1 1 0 1 1 1 1 0 0 0 0 1 0 0 1 0 1 0 0 1 1 0 1. We found that the measure most driven by description length was compressibility.

The longest description of a binary sequence in the OEIS, identified as A123594, reads ``Unique sequence of 0s and 1s which are either repeated or not repeated with the following property: when the sequence is `coded' in writing down a 1 when an element is repeated and a 0 when it is not repeated and by putting the initial element in front of the sequence thus obtained, the above sequence appears".

\subsection{Results}

We found that the textual description length as derived from the database is, as illustrated above, best correlated with the AP-based (BDM) measure, with Spearman test statistic 0.193, followed by compression (only the sequence is compressed, not the description) with 0.17, followed by entropy, with 0.09 (Fig.~\ref{fig:main}). Spearman rank correlation values among complexity measures reveal how these measures are related to each other with BDM v Compress: 0.21, BDM v Entropy: 0.029 and Compress v Entropy: -0.01 from 875 binary sequences in the OEIS database.

\begin{figure}[htbp!]
  \centering

  \includegraphics[width=5.3cm]{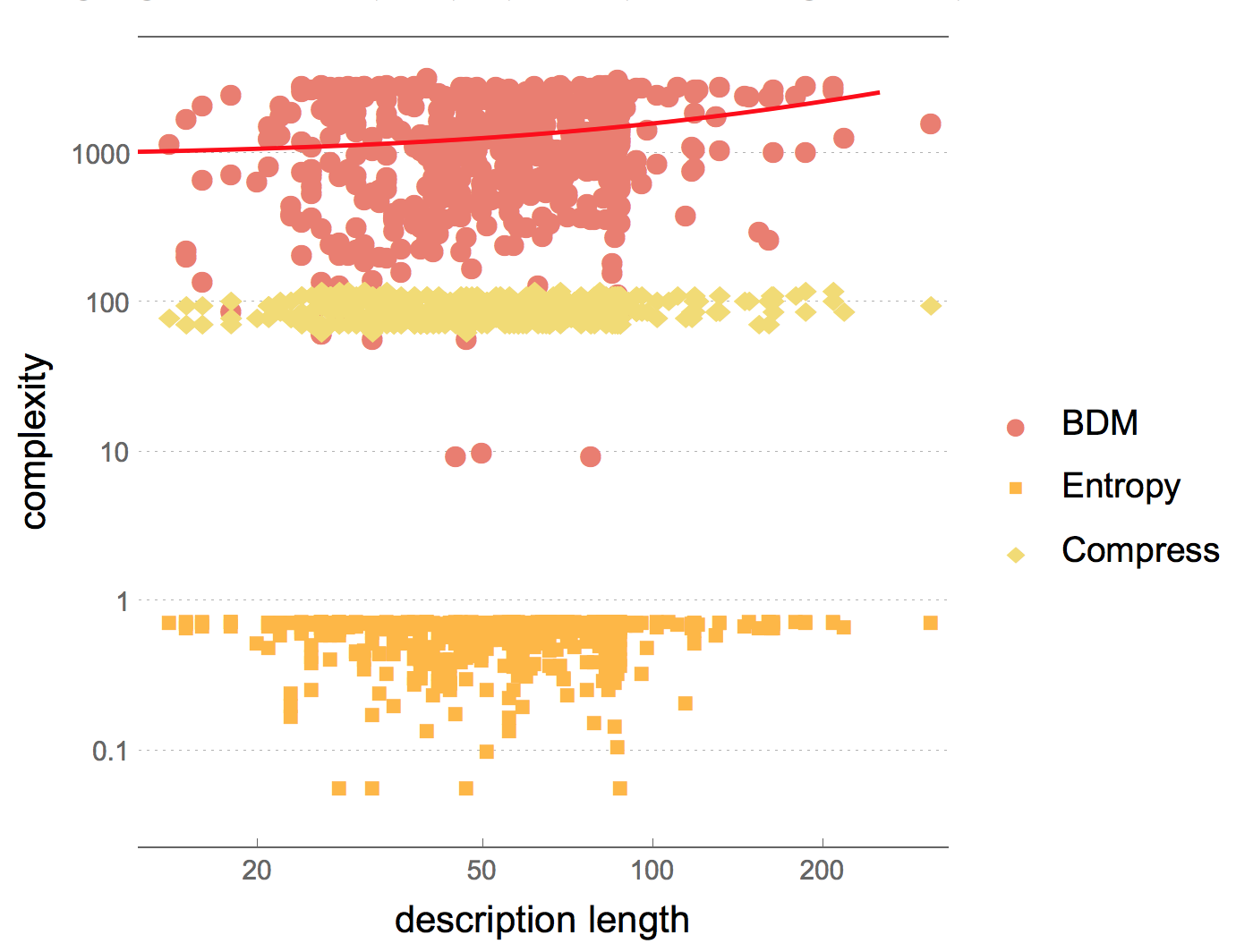}\hspace{.5cm}  \includegraphics[width=6.2cm]{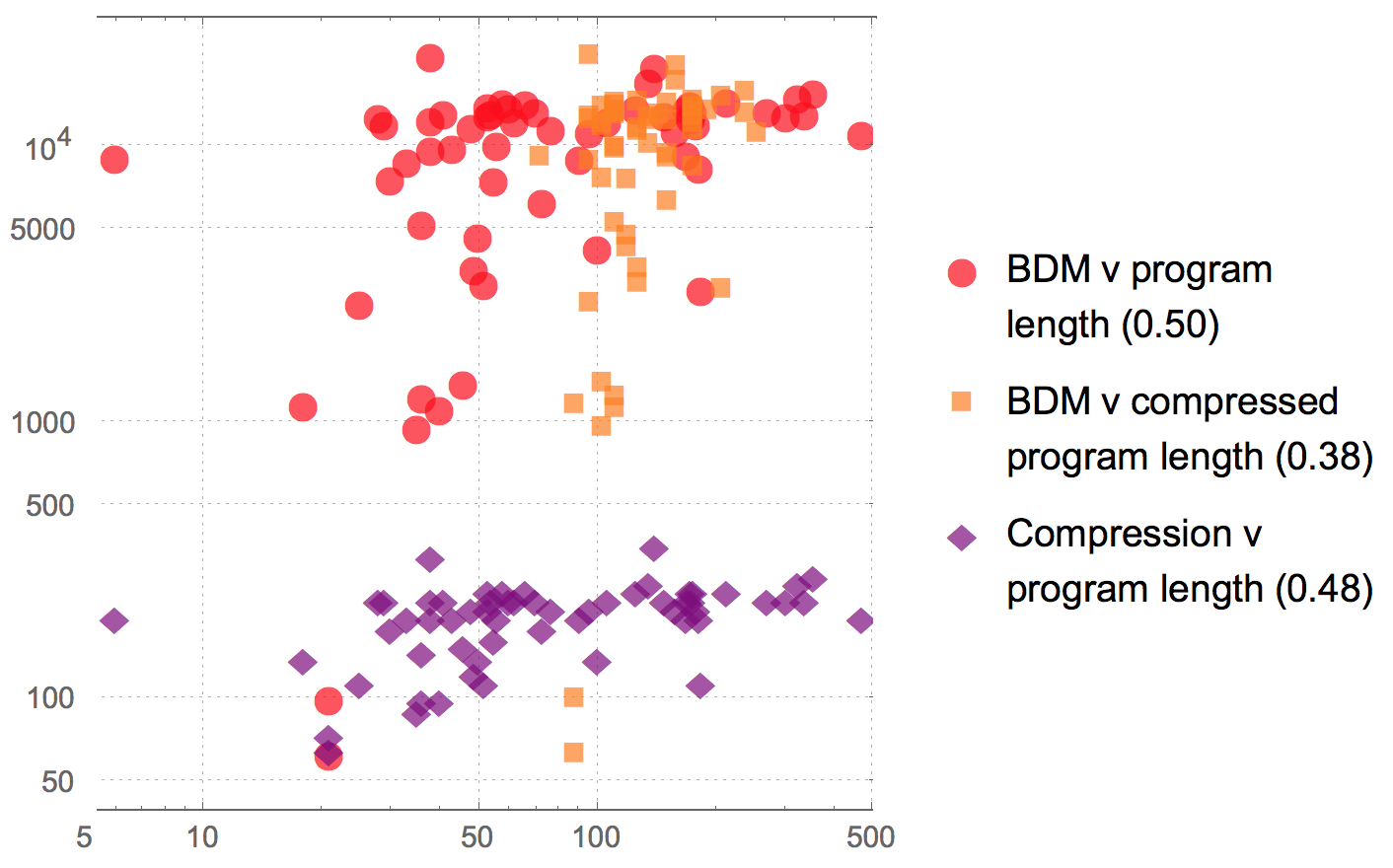}
  \caption{Left: Correlation between the estimated algorithmic complexity ($\log$) by the AP-based measure (BDM) an the length of the text description of each sequence from the OEIS. Fitted line for highest correlation (BDM) is given by $955.709 + 6.47818 x$ using least squares. Right: Algorithmic complexity estimation by BDM ($\log$) and of compression on program length (in the Wolfram Language/Mathematica) as coming from the OEIS. In parenthesis the Spearman rank correlation values for each case. Further compressing the program length using Compress resulted in a lower correlation value and BDM outperformed lossless compression.}
  \label{fig:main}
\end{figure}

We noticed that the descriptions of some sequences referred to other sequences to produce a new one (e.g. ``A051066 read mod 2"). This artificially made some sequence descriptions look shorter than they should be. When avoiding all sequences referencing others, all Spearman rank values increased significantly, with values 0.25, 0.22 and 0.12 for BDM, compression and entropy respectively. 

To test whether the AP-based (BDM) measure captures some algorithmic content that the best statistical measures (Compress and entropy) may be missing, we compressed the sequence description and compared again against the sequence complexity. The correlation between the compressed description and the sequence compression came closer to that of the AP-estimation by BDM, and BDM itself was even better. The Spearman values after compressing textual descriptions were 0.27, 0.24 and 0.13 for BDM, Compress and entropy respectively.

We then looked at 139\,546 integer sequences from the OEIS database, avoiding other non-integer sequences in the database. Those considered represent more than half of the database. Every integer was converted into binary, and for each binary sequence representing an integer an estimation of its algorithmic complexity was calculated. We compared the total sum of the complexity of the sequence (first 40 terms) against its text description length (both compressed and uncompressed) by converting every character into its ASCII code, program length and function lengths, these latter in the Wolfram Language (using Mathematica). While none of those descriptions can be considered the shortest possible, their lengths are upper bounds of the maximum possible lengths of the shortest versions. As shown in Fig.~\ref{fig:main}, we found that the AP-based measure (BDM) performed best when comparing program size and estimated complexity from the program-generated sequence.

\section{Conclusion}
\label{sec:comm}

Computable approximations to algorithmic information measures are certainly useful. For example, lossless compression methods have been widely used to approximate $K$, despite their limitations and their departure from algorithmic complexity. Most of these algorithms are in fact entropy-rate estimators~\cite{zenildata} e.g. all those based in LZ and LZW algorithms such as zip, gzip and png. In this paper we have studied the formal properties of a computable algorithmic
probability measure $m$ and of finite approximations $m_k$ to
$m$. These measures can be used to approximate $K$ by means of the Coding Theorem Method (CTM), despite the \textit{invariance theorem}, that sheds no light on the rate of convergence to $K$. Here we compared $m$ and $D(5)$ and concluded that for practical purposes the two produce similar results. What we have reported in this paper are the first steps toward a formal analysis of finite approximations to algorithmic
probability-based measures based on small Turing machines. The results shown in Fig.~\ref{fig:main} strongly suggest that AP-based measures are not only an alternative to lossless compression algorithms for estimating algorithmic (Kolmogorov-Chaitin) complexity, but may actually capture features that statistical methods such as lossless compression, based on popular algorithms such as LWZ and entropy, cannot capture.

\section*{Acknowledgments}
\label{sec:acknowledgments}

The authors wish to thank the rest of the Algorithmic Nature lab.

\end{document}